\documentclass[a4paper,USenglish,10pt]{article}
\usepackage[utf8]{inputenc}

\usepackage{amsthm,amssymb,amsfonts,amsmath}
\usepackage{amssymb,amsfonts,amsmath}
\newtheorem{theorem}{Theorem}
\newtheorem{lemma}[theorem]{Lemma}

\newtheorem{observation}[theorem]{Observation}
\theoremstyle{definition}
\newtheorem{definition}[theorem]{Definition}

\usepackage[hidelinks]{hyperref}
\usepackage[blocks]{authblk}

\usepackage{graphicx}
\graphicspath{{figures/}}
\usepackage{wrapfig}
\usepackage{caption,subcaption}

\makeatletter
\g@addto@macro\bfseries{\boldmath}
\makeatother

\advance\hoffset-7mm
\usepackage{thm-restate}

\usepackage{url}

\newcommand{\A}{\ensuremath{\mathcal{A}}}
\newcommand{\B}{\mathcal{B}}
\newcommand{\F}{\mathcal{F}}
\newcommand{\s}{\mathcal{S}}
\newcommand{\eps}{\varepsilon}

\newcommand{\R}{\ensuremath{\mathbb{R}}}

\makeatletter
\g@addto@macro\bfseries{\boldmath}
\makeatother

\usepackage{authblk}

\date{}

\title{Weighted $\eps$-Nets}

\begin{document}
\author[1]{Daniel Bertschinger}
\affil[1]{Department of Computer Science, ETH Z\"urich.\\ \texttt{daniel.bertschinger@inf.ethz.ch}}

\author[2]{Patrick Schnider}
\affil[2]{Department of Computer Science, ETH Z\"urich.\\ \texttt{patrick.schnider@inf.ethz.ch}}

\maketitle

\begin{abstract}
Motivated by recent work of Bukh and Nivasch \cite{Bukh} on one-sided $\eps$-approximants, we introduce the notion of \emph{weighted $\eps$-nets}. It is a geometric notion of approximation for point sets in $\R^d$ similar to $\eps$-nets and $\eps$-approximations, where it is stronger than the former and weaker than the latter. The main idea is that small sets can contain many points, whereas large sets must contain many points of the weighted $\eps$-net. 

In this paper, we analyze weak weighted $\eps$-nets with respect to convex sets and axis-parallel boxes and give upper and lower bounds on $\eps$ for weighted $\eps$-nets of size two and three. Some of these bounds apply to classical $\eps$-nets as well.
\end{abstract}\section{Introduction}

Representing large, complicated objects by smaller, simpler ones is a common theme in mathematics. For one-dimensional data sets this is realized by the notions of medians, means and quantiles. One fundamental difference between medians and quantiles on the one side and the mean on the other side is the robustness of the former against outliners of the data. 

\subparagraph{Centerpoint.}
Medians and quantiles are one-dimensional concepts, whereas modern data sets are often multidimensional. Hence, many generalizations of medians and quantiles to higher dimensions have been introduced and studied. One example is the notion of a centerpoint, that is, a point $c$ such that for every closed halfspace $h$ containing $c$ we know that $h$ contains at least a $\frac{1}{d+1}$-fraction of the whole data, where $d$ denotes the dimension. The Centerpoint Theorem ensures that for any point set in $\R^d$ there always exists such a centerpoint \cite{RadoCenterpoint}. 

Instead of representing a data set by a single point, one could take a different point set as a representative. This is exactly the idea of an \emph{$\eps$-net}. 
\begin{definition}
Given any range space $(X, \mathcal{R})$, an \emph{$\eps$-net} on a point set $P \subseteq X$ is a subset $N \subseteq P$ such that every $R \in \mathcal{R}$ with $|R \cap P| \geq \eps |P|$ has nonempty intersection with $N$. If the condition that an $\eps$-net needs to be a subset of $P$ is dropped, then $N$ is called a \emph{weak $\eps$-net}.
\end{definition}

In this language, a centerpoint is a weak $\frac{d}{d+1}$-net for the range space of halfspaces. The concept of $\eps$-nets has been studied in a huge variety; first, there are statements on the existence and the size of $\eps$-nets, if $\eps$ is given beforehand. On the other hand, one can fix the size of the $\eps$-net a priori and try to bound the range of $\eps$ in which there always exists an $\eps$-net. For the former, it is known that every range space of VC-dimension $\delta$ has an $\eps$-net of size at most $\mathcal{O}(\frac{\delta}{\eps} \log \frac{1}{\eps})$ \cite{Welzlnets}.

\subparagraph{$\eps$-Approximations.}
For some applications though, $\eps$-nets may not retain enough information. For every range we only know that it has a nonempty intersection with the net; however, we do not know anything about the size of this intersection. Hence, the following definition of $\eps$-approximations comes naturally. 

\begin{definition}
Given any range space $(X,\mathcal{R})$ and any parameter $0 \leq \eps \leq 1$, an \emph{$\eps$-approximation} on a point set $P \subset X$ is a subset $A \subset P$ such that for every $R \in \mathcal{R}$ we have $\left| \frac{|R \cap P|}{|P|} - \frac{|R \cap A|}{|A|} \right| \leq \eps.$
\end{definition}

Initiated by the work of Vapnik and Chervonenkis \cite{VapnikIdeanets}, one general idea is to construct $\eps$-approximations by uniformly sampling a random subset $A \subseteq X$ of large enough size. This results in statements about the existence of $\eps$-approximations depending on the VC-dimension of the range space. In particular every range space of VC-dimension $\delta$ allows an $\eps$-approximation of size $\mathcal{O}(\frac{\delta}{\eps^2} \log \frac{1}{\eps})$ \cite{HarpeledSurvey, HarPeledSharir, Welzlapproximation}.

\subparagraph{Convex Sets.}
It is well-known that the range space of convex sets has unbounded VC-dimension; therefore, none of the results mentioned above can be applied. While constant size weak $\eps$-nets still exist for the range space of convex sets \cite{alonbaranyconvexsets, Rubinconvexsets}, the same cannot be said for weak $\eps$-approximations (Proposition 1 in \cite{Bukh}). Motivated by this, Bukh and Nivasch \cite{Bukh} introduced the notion of \emph{one-sided weak $\eps$-approximants}. The main idea is that small sets can contain many points, whereas large sets must contain many points of the approximation. Bukh and Nivasch show that constant size one-sided weak $\eps$-approximants exist for the range space of convex sets. In this work, we define a similar concept, called \emph{weighted $\eps$-nets}. In contrast to one-sided weak $\eps$-approximants, our focus is to understand what bounds can be achieved for a fixed small value of $k$, which is given a priori. In this sense our approach is similar to the one taken by Aronov et al.~\cite{AronovHalfplanes2D} (for standard $\eps$-nets).

\begin{definition}
Given any point set $P \subset \R^d$ of size $n$, a \emph{weighted $\eps$-net of size $k$} (with respect to some range space) is defined as a set of points $p_1, \ldots, p_k$ and some values $\eps = (\eps_1, \ldots, \eps_k)$ such that every set in the range space containing more than $\eps_i n$ points of $P$ contains at least $i$ of the points $p_1, \ldots, p_k$. \end{definition}

Following historic conventions, we denote a weighted $\eps$-net as \emph{strong} if $p_1,\ldots,p_k \in P$ and as \emph{weak} otherwise. In this work, we focus on weak weighted $\eps$-nets of small size for the range space of convex sets and axis-parallel boxes.

\section{Weighted $\eps$-nets for the range space of convex sets}
\label{sec:convex}

Weighted $\eps$-nets for the range space of halfspaces were already studied by Pilz and Schnider \cite{Pilz}. In this section we generalize one of their results to the range space of convex sets.

\begin{theorem}
\label{theorem:convex}
Let $P$ be a set of $n$ points in general position in $\R^d$. Let $0 < \eps_1 \leq \eps_2 < 1$ be arbitrary constants with $(i)$ $d \eps_1 + \eps_2 \geq d$ and $(ii)$ $\eps_1 \geq \frac{2d-1}{2d+1}$. Then there are two points $p_1$ and $p_2$ in $\mathbb{R}^d$ such that
\begin{itemize}
\item[1.] every convex set containing more than $\eps_1 n$ points of $P$ contains at least one of the points $p_1$ or $p_2$, and
\item[2.] every convex set containing more than $\eps_2 n$ points of $P$ contains both $p_1$ and $p_2$.
\end{itemize}
\end{theorem}

\begin{proof}
The main idea of the proof is to create two classes $\A$ and $\B$ containing convex subsets of $\R^d$ and to show that all sets of the same class have a common intersection.

Every convex subset of $\R^d$ containing more than $\eps_2 n$ points of $P$ is put into both, $\A$ and $\B$. From now on, we will call such sets \emph{big sets}. Similarly, let us denote by \emph{small sets}, the convex subsets of $\R^d$ containing more than $\eps_1 n$, but at most $\eps_2 n$ points of $P$. Let $H$ be a $(d-1)$-dimensional hyperplane separating the point set into two disjoint, equally sized subsets; the halfspace \emph{above} $H$, and the halfspace \emph{below} $H$. Every small set containing more points of $P$ below $H$ than above $H$ is put into $\B$. Every small set which is not in $\B$, is put into $\A$. 

We now show that any $d+1$ sets in $\A$ and in $\B$, respectively, have non-empty intersection. This will allow us to use Helly's Theorem and conclude the proof. 

\begin{lemma}
\label{lemma:onebig}
Let $0 \leq k \leq d$ be any integer. The intersection $\mathcal{I}$ of any $k$ small sets with any $(d+1-k)$ big sets is non-empty. Moreover, it contains at least one point of $P$.
\end{lemma}

\begin{proof}
This observation follows from a simple counting argument. The complement of any big set contains fewer than $(1 - \eps_2)n$ points of $P$. Hence, there are strictly fewer than $2 (1 - \eps_2)n$ points of $P$ not in the intersection of two big sets. Stated equivalently, there are strictly more than $n - 2 (1- \eps_2)n$ points of $P$ in the intersection of any two big sets. By the exact same reasons, there are strictly more than $n - l (1 - \eps_2)n$ points in the intersection of any $l$ big sets. Consequently, the intersection $\mathcal{I}$ of $k$ small sets and $l$ big sets contains strictly more than $n - k(1- \eps_1)n - l(1 - \eps_2)n$ points of $P$. 

Recall that $d \eps_1 + \eps_2 \geq d$ holds by assumption $(i)$ of the Theorem. Since $\eps_1 \leq \eps_2$ it follows that $(d-k') \eps_1 + ( 1 + k' ) \eps_2 \geq d$ holds for any $0 \leq k' \leq d$. We have
\begin{align*}
|\mathcal{I} \cap P| &> n - k (1 - \eps_1)n - (d + 1 - k)(1 - \eps_2)n \\ 
&= n + \Big( k \eps_1 + \big(1 + (d-k)\big) \eps_2 \Big) n - \Big( k + (d + 1 - k) \Big) n \\
&=n + \Big( (d - k') \eps_1 + (1 + k') \eps_2 \Big) n - (d + 1) n \\
&\geq n + dn - (d+1)n = 0,
\end{align*}
and therefore this intersection contains at least one point of $P$, which proves the Lemma.
\end{proof}

\begin{lemma}
\label{lemma:onlysmall}
Any $d + 1$ small sets in $\A$ (and in $\B$, respectively) have a common, non-empty intersection. 
\end{lemma}

\begin{proof}
Let $A_1, \ldots, A_{d+1}$ be any $d+1$ small sets of $\A$ and assume for the sake of contradiction that they do not have any common intersection. Any $d$ of them have nonempty intersection by a similar counting argument to the one seen above. Therefore, define \[B_i := \bigcap_{\substack{k=1 \\ k \neq i}}^{d+1}{A_k},\] and observe the following.

\begin{observation}
\label{claim:hyperplane}
For every $(d-1)$-dimensional hyperplane $H'$ in $\R^d$ there is one set $B_j$ that does not intersect $H'$. 
\end{observation}

\begin{proof}
This observation can for example be seen by a contradiction argument. Assume that every $B_j$ has an intersection with $H'$. Then the family $A_1 \cap H', \ldots, A_{d+1} \cap H'$ of $(d-1)$-dimensional convex sets satisfies the Helly condition and there exists a point in the intersection of all these sets. This contradicts the assumption that $\bigcap_{i=1}^{d+1}{A_i}$ is empty.
\end{proof}

Following Observation \ref{claim:hyperplane} not all the intersections $B_j$ can intersect the hyperplane $H$, chosen at the very beginning in the proof of Theorem \ref{theorem:convex}. Hence one intersection, without loss of generality let us assume that this is $B_1$, lies completely on one side of $H$.

Assume for now that $B_1$ lies below $H$. However, $A_2, \ldots, A_{d+1}$ all lie in $\A$, so these small sets contain at least the same number of points of $P$ above $H$, as they contain points of $P$ below $H$. Additionally the parts of $A_2, \ldots, A_{d+1}$ lying above $H$ are disjoint by assumption. Hence, the number of points of $P$ above $H$ is at least \[d \frac{\eps_1}{2} n > \frac{n}{2},\] assuming $d \geq 2$ and therefore $d \eps_1 \geq 1$. This strict inequality contradicts the condition that $H$ is a halving plane. Thus $B_1$ lies above $H$. 

Since in this case, $A_1$ and $B_1$ are disjoint, the number of points above $H$ is at least 
\begin{align*}
\frac{|A_1|}{2} + |B_1| &> \frac{\eps_1}{2}n + \big( n - d(1 - \eps_1) n \big) \\ 
&= n \cdot \Big( 1 - d + d \eps_1 + \frac{\eps_1}{2}\Big) \\
&\geq n \cdot \bigg( 1 - d + d \, \frac{2d-1}{2d+1} + \frac{2d-1}{2(2d+1)} \bigg) = \frac{n}{2},
\end{align*}
where we used inequality $(ii)$ of the Theorem. This strict inequality is again a contradiction to the fact that $H$ is a halving hyperplane, which concludes the proof of Lemma \ref{lemma:onlysmall}.
\end{proof}

Together, Lemmas \ref{lemma:onebig} and \ref{lemma:onlysmall} show that any $d + 1$ sets of $\A$ and $\B$ respectively, have a common intersection. Hence by Helly's Theorem, there exist points $p$ and $p'$ in the intersection of all sets in $\A$ and $\B$. Choosing $p_1 := p$ and $p_2 := p'$ concludes the proof of the Theorem.
\end{proof}

\section{Lower Bounds on $\eps$}
\label{sec:lowerbounds}

Having seen an existential result for weighted $\eps$-nets with respect to convex sets, we are interested in the best possible value for $\eps$. In this chapter we present some lower bounds on $\eps$. First, an example given in \cite{Pilz} can be adapted to show that inequality $(i)$ of Theorem \ref{theorem:convex} is needed in the following sense: In the plane we cannot simultaneously have $\eps_1 > \frac{3}{5}$ and $\eps_2 > \frac{4}{5}$. To see this, consider the point set in Figure \ref{fig:inequality(i)}. Note that one of the two points needs to lie in $l_{a,d}^+ \cap l_{a,d}^-$. The same is true for all intersections depicted in the right part of Figure \ref{fig:inequality(i)}. However, these five intersections cannot be stabbed using only two points. 

\begin{figure}[h]
\centering
	\begin{subfigure}{.49\linewidth}
		\centering
		\includegraphics[scale = 1]{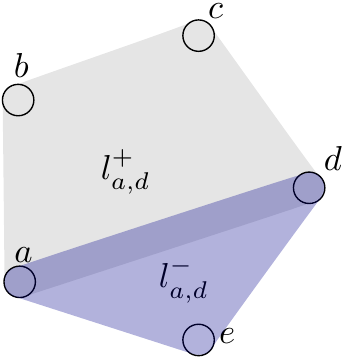}
	\begin{minipage}{.1cm}
	\vfill
	\end{minipage}
        \end{subfigure}
        \hfill
	\begin{subfigure}{.49\linewidth}
		\centering
		\includegraphics[scale = 1]{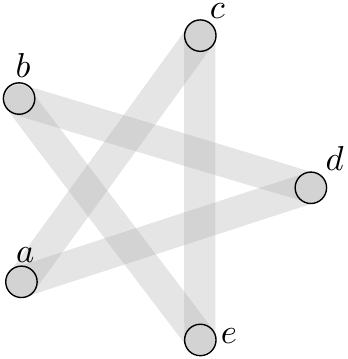}
	\end{subfigure}
	\caption{A point set of five regions in convex position, each containing exactly $k$ points. Two particular regions containing four (three, respectively) of the regions. The intersections of interest are drawn on the right side.}
	\label{fig:inequality(i)}
\end{figure}

\subsection{Lower bounds on $\eps_1$}

On the other hand, one can give lower bounds on $\eps_1$, independently of the value of $\eps_2$. This setting is exactly the same as giving lower bounds on $\eps$ for any $\eps$-net. Hence, any bound given in this chapter is also a lower bound on $\eps$ for $\eps$-nets as well. Mustafa and Ray \cite{MustafaConvexSets2D} have studied this in dimension $2$, showing that there exist point sets $P$ in $\R^2$ such that for every two points $p_1$ and $p_2$, not necessarily in $P$, we can find a convex set containing at least $\frac{4n}{7}$ points of $P$ but neither $p_1$ nor $p_2$.

For higher dimension, to our knowledge the bounds given here are among the first and currently the best lower bounds for the range space of convex sets. 

\begin{lemma}
\label{Claim:alphain3D}
There are point sets $P \subset \R^3$, such that for any two points $p_1$ and $p_2$ in $\R^3$ we can always find a compact convex set containing at least $\frac{5n}{8}$ points of $P$, but neither $p_1$ nor $p_2$.
\end{lemma}

\begin{proof}
Consider the following point set in three dimensions consisting of eight points. There is a hexagon in the $xy$-plane, one point above the hexagon (denoted as $u_1$), and one point below the hexagon (denoted as $u_2$), see Figure \ref{fig:alphain3D3D}. 
\begin{figure}[h]
\centering
	\begin{subfigure}{.32\linewidth}
		\centering
		\includegraphics[scale = 1]{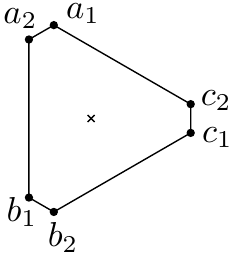}
	\begin{minipage}{.1cm}
	\vfill
	\end{minipage}
        \end{subfigure}
        \hfill
	\begin{subfigure}{.32\linewidth}
		\centering
		\includegraphics[scale = 1]{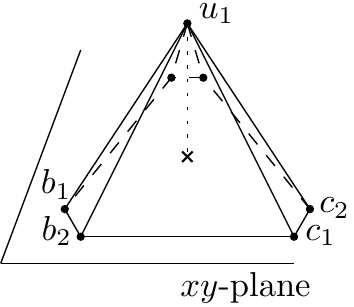}
	\begin{minipage}{.1cm}
	\vfill
	\end{minipage}		
	\end{subfigure}
        \hfill
	\begin{subfigure}{.32\linewidth}
		\centering
		\includegraphics[scale = 1]{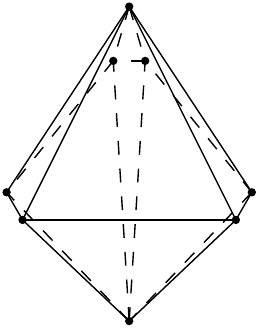}
	\begin{minipage}{.1cm}
	\vfill
	\end{minipage}
	\end{subfigure}
	\caption{A point set in three dimensions, with six points in the $xy$-plane arranged in a hexagon, one point above the $xy$-plane and one point below the $xy$-plane.}
	\label{fig:alphain3D3D}
\end{figure}

We prove that for any two points $p_1, p_2 \in \R^3$ there exists a convex set containing at least $5$ points of $P$, but neither $p_1$ nor $p_2$, which then concludes the proof of Lemma \ref{Claim:alphain3D}. 
First, if none of $p_1$ and $p_2$ lies in the $xy$-plane we are done directly; hence, let us assume that at least one of them, without loss of generality $p_1$, lies in the $xy$-plane. 

\begin{observation}
If $p_2$ does not lie in the $xy$-plane, then the Lemma is true. 
\end{observation}

\begin{proof}
Consider sets of four consecutive vertices of the hexagon in the $xy$-plane in the form of the colored areas of the first part of Figure \ref{fig:alphain3D}. At least one of these areas does not contain $p_1$. Whenever $p_2$ lies above the $xy$-plane, then the mentioned area together with $u_2$ forms a set of $5$ points of $P$ but it contains neither $p_1$ nor $p_2$. Similarly, if $p_2$ lies below the plane, we simply replace $u_2$ by $u_1$.
\end{proof}

\begin{figure}[h]
\centering
	\begin{subfigure}{.32\linewidth}
		\centering
		\includegraphics[scale = 1]{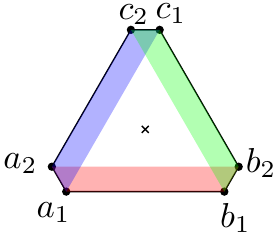}
	\begin{minipage}{.1cm}
	\vfill
	\end{minipage}
        \end{subfigure}
        \hfill
	\begin{subfigure}{.32\linewidth}
		\centering
		\includegraphics[scale = 1]{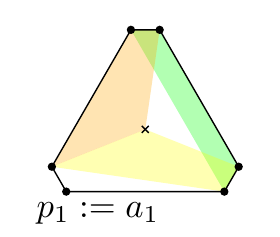}
	\begin{minipage}{.1cm}
	\vfill
	\end{minipage}
	\end{subfigure}
        \hfill
	\begin{subfigure}{.32\linewidth}
		\centering
		\includegraphics[scale = 1]{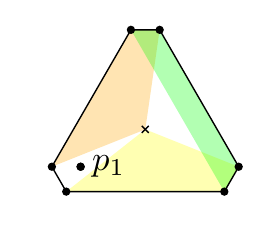}
	\begin{minipage}{.1cm}
	\vfill
	\end{minipage}
	\end{subfigure}	
	\caption{Projections of certain subsets of $P$ onto the $xy$-plane. The cross indicates the projection of $u_1$ and $u_2$.}
	\label{fig:alphain3D}
\end{figure}

Consequently, let us assume that both $p_1$ and $p_2$ lie in the $xy$-plane. At least one of $p_1$ and $p_2$, without loss of generality $p_1$, lies in two colored $4$-gons simultaneously. Let us assume that these are the red and blue areas.

If $p_1$ is equal to a point of the hexagon, without loss of generality $p_1 := a_1$, then consider the sets $A := \left\{ b_1, b_2, c_1, c_2, u_1\right\}$, $B := \left\{ a_2, b_1, b_2, u_1, u_2\right\}$ and $C := \left\{ a_2, c_1, c_2, u_1, u_2\right\}$. The projections of their convex hulls are drawn in the middle part of Figure \ref{fig:alphain3D}. None of $A,B$ and $C$ contains $p_1$ and they do not have a common intersection. Hence, one of them does not contain $p_2$ and we are done. 

If $p_1$ is not equal to a point of the hexagon, then define $A' := \left\{ b_1, b_2, c_1, c_2, u_1\right\}$, $B' := \left\{ a_1, b_1, b_2, u_1, u_2\right\}$ and $C' := \left\{ a_2, c_1, c_2, u_1, u_2\right\}$, see the rightmost part of Figure \ref{fig:alphain3D}. Again, none of $A,B$ and $C$ contains $p_1$ and they do not have a common intersection. Hence, one of them does not contain $p_2$ and we are done. This concludes the proof of Lemma \ref{Claim:alphain3D}.
\end{proof}


For general dimensions a lower bound on $\eps_1$ is given in the following Theorem. 

\begin{theorem}
\label{Lemma:alphain4D}
There are point sets $P$ in $\R^d$ such that for any two points $p_1, p_2 \in \R^d$ there is a compact convex set containing $\frac{d}{d+2}$ of the points of $P$, but neither $p_1$ nor $p_2$.
\end{theorem}

\begin{figure}[h]
\centering
	\begin{subfigure}{.32\linewidth}
		\centering
		\includegraphics[scale = 1]{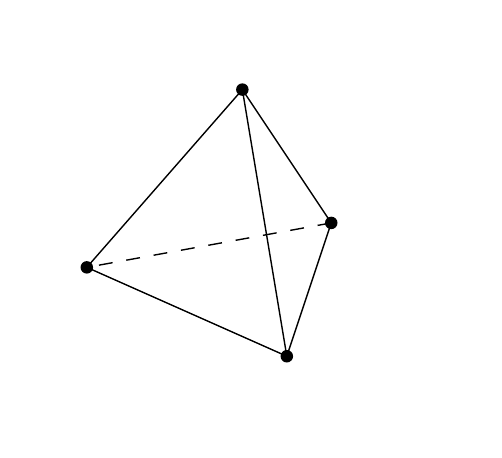}
	\begin{minipage}{.1cm}
	\vfill
	\end{minipage}
        \end{subfigure}
        \hfill
	\begin{subfigure}{.32\linewidth}
		\centering
		\includegraphics[scale = 1]{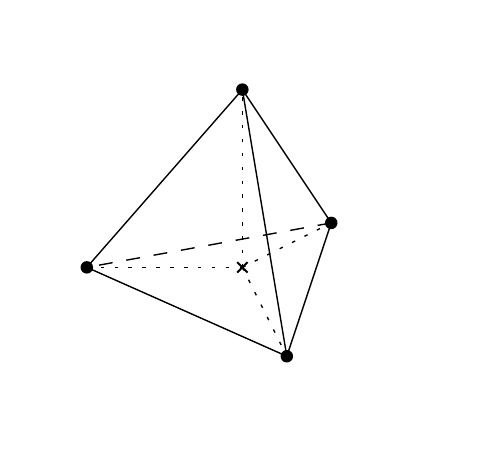}
	\begin{minipage}{.1cm}
	\vfill
	\end{minipage}
        \end{subfigure}
                \hfill
	\begin{subfigure}{.32\linewidth}
		\centering
		\includegraphics[scale = 1]{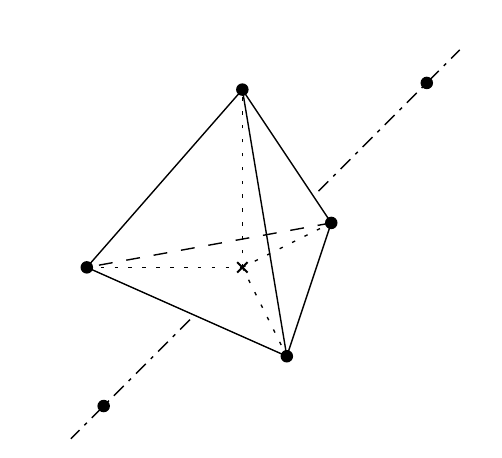}
	\begin{minipage}{.1cm}
	\vfill
	\end{minipage}
        \end{subfigure}
	\caption{We start with a $(d-1)$-simplex, in this case a $3$-simplex (or tetrahedron), and we pierce a line through the center of the tetrahedron. Once the intersection of the line and the simplex is indicated and once the line and the additional points above and below the simplex are drawn.}
	\label{fig:alphain4D}
\end{figure}
For the proof of this Theorem, consider the following construction of a point set $P \subset \R^d$, consisting of exactly $d+2$ points. First, take a $(d-1)$-simplex $\s$ that consists of $d$ points, denoted as $v_1, v_2, \ldots, v_d$. Any such simplex $\s$ can not use the whole $d$ dimensions by definition. Therefore we can pierce a (one-dimensional) line through the center of $\s$ and fix two points on this line, that is, one on each side, denoted as the point above the simplex, $u_1$, and the point below $\s$, $u_2$. Figure \ref{fig:alphain4D} should give an idea of the construction in four-dimensional space.

We prove the following Lemma, which directly implies Theorem \ref{Lemma:alphain4D}.

\begin{lemma}
\label{Lemma:alphain4Db}
For any two points $p_1, p_2 \in \R^d$ there exists a compact convex set containing $d$ of the points of $P$, but neither $p_1$ nor $p_2$.
\end{lemma}

\begin{proof}
We consider several cases depending on the position of $p_1$. In any case, we can not find a point $p_2 \in \R^d$ such that every compact convex set containing $d$ points of $P$ contains at least one of them. 

\begin{observation}
\label{obs:Simplex01}
If $p_1$ lies below (or above) $\s$, then there exists a compact convex set containing $d$ points of $P$, but neither $p_1$ nor $p_2$. 
\end{observation}

\begin{proof}
Assume that $p_1$ lies strictly below $conv(\s)$. Define the sets \[C_i := conv \Big( \big( \s \smallsetminus \{ v_i \} \big) \cup \{ u_1 \} \Big).\] These convex sets contain $d$ points each, but none of them contains $p_1$, as they do not contain any point below $\s$. Therefore all of them have to contain $p_2$, otherwise we are done. The intersection $\bigcap_{i=1}^{d}{C_i}$ lies strictly above $conv(\s)$; thus, $p_2$ has to lie above $conv(\s)$. However, in this case $conv(\s)$ contains neither $p_1$ nor $p_2$.
\end{proof}

\begin{observation}
\label{obs:Simplex02}
Let $p_1$ be a vertex of $\s$, without loss of generality $p_1 := v_1$. There exists a compact convex set containing $d$ points of $P$, but neither $p_1$ nor $p_2$. 
\end{observation}

\begin{proof}
Clearly $p_2$ has to lie in the opposite $(d-2)$-simplex $\s' \subseteq \s$, where opposite means the simplex formed by all vertices of $\s$ except $v_1$. This is because $p_2$ has to lie in the following convex sets: 
\begin{align*}
D_1&:=conv\Big( \big( \s \smallsetminus \{ v_1 \} \big) \cup \{ u_1 \} \Big) \text{ and} \\
D_2&:=conv\Big( \big( \s \smallsetminus \{ v_1 \} \big) \cup \{ u_2 \} \Big),
\end{align*}
whose intersection is exactly $\s'$. Consider \[E_i := conv \Big( \big( \s \smallsetminus \{ v_1, v_i \} \big) \cup \{ u_1, u_2\} \Big),\] containing $d$ points as well; while none of them contains $p_1$. Now $\bigcap_{i=1}^{d+1}{E_i} \, \cap \, conv(\s)$ is empty. Thus at least one of these sets contains neither $p_1$ nor $p_2$.
\end{proof}

\begin{observation}
\label{obs:Simplex03}
Let $p_1$ and $p_2$ be convex combinations of more than one vertex of $\s$. Then there exists a compact convex set containing $d$ points of $P$, but neither $p_1$ nor $p_2$. 
\end{observation}

\begin{proof}
If $p_1$ is a convex combination of more than one vertex of $\s$ and $p_2$ is anything else but a convex combination of more than one vertex of $\s$, we can rename the points and use one of the Lemmas above. Hence, let $p_1$ be a convex combination of $k \leq d$ vertices of the simplex; without loss of generality these are $v_1, v_2, \ldots, v_k$, and let $p_2$ be a convex combination of $k' \leq d$ vertices of the simplex, without loss of generality $v_{k+1}, \ldots, v_{k + k'}$. Define \[A := conv\Big( \{v_2, v_3, \ldots, v_k\} \cup \{ v_{k+2}, \ldots, v_{k + k'} \} \cup \{ v_{k + k'+1}, \ldots, v_d, u_1, u_2 \}\Big).\] It contains every vertex of the convex combination of $p_1$ except one, every vertex of the combination of $p_2$ except one and every other vertex. Hence altogether, $A$ contains $d$ points of $P$. However, $A$ can not contain $p_1$ and it can not contain $p_2$ by construction. 
\end{proof}

Observations \ref{obs:Simplex01}, \ref{obs:Simplex02} and \ref{obs:Simplex03} together prove Lemma \ref{Lemma:alphain4Db} as these are the only situations that can possibly occur. 
\end{proof}

\section{The range space of axis-parallel boxes}

In this section, we study weighted $\eps$-nets of size $2$ and $3$ for the range space of axis-parallel boxes. Axis-parallel boxes have the property that they allow a much stronger Helly-type result.

\begin{observation}
\label{prop:boxes}
Let $\F$ be a family of compact, axis-parallel boxes in $\R^d$ such that any two of them have a common intersection. Then the whole collection has a nonempty intersection.
\end{observation}

As a direct consequence of this observation we note that for any point set $P$ in $\R^d$, there always exists a (weighted) $\frac{1}{2}$-net of size $1$ for the range space of axis-parallel boxes. For weighted $\eps$-nets of larger size we find the following.

\begin{theorem}
\label{theorem:boxes2}
Let $P$ be a set of $n$ points in general position in $\R^d$. Let $0 < \eps_1 \leq \eps_2 < 1$ be arbitrary constants with $(i)$ $\eps_1 \geq \frac{3^{d-1}}{2 \cdot 3^{d-1} + 1}$ and $(ii)$ $\eps_1 + \eps_2 \geq 1$. Then there exist two points $p_1$ and $p_2$ such that 
\begin{itemize}
\item[1.] every axis-parallel box containing more than $\eps_1 n$ points of $P$ contains at least one of the points $p_1$ and $p_2$, and
\item[2.] every axis-parallel box containing more than $\eps_2 n$ points of $P$ contains both, $p_1$ and $p_2$.
\end{itemize}
\end{theorem}

\begin{proof}[Sketch of Proof]
For the sake of simplicity, we only present a proof in $\R^2$ with fixed values $\eps_1 = \frac{3}{7}$ and $\eps_2 = \frac{4}{7}$. For other values the proof works analogously. First, divide the point set with a horizontal line $l_1$, such that there are $\frac{3n}{7}$ points below $l_1$ and $\frac{4n}{7}$ points above $l_1$. Then add two lines $l',l''$ perpendicular to $l_1$ splitting the point set below $l_1$ into three parts containing the same number of points, see Figure \ref{fig:boxes2punktein2D} (left).
\begin{figure}[h]
\centering
	\begin{subfigure}{.24\linewidth}
		\centering
		\includegraphics[scale = 1]{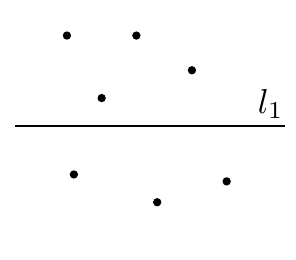}
	\begin{minipage}{.1cm}
	\vfill
	\end{minipage}
        \end{subfigure}
        \hfill
	\begin{subfigure}{.24\linewidth}
		\centering
		\includegraphics[scale = 1]{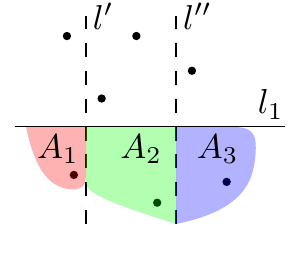}
	\begin{minipage}{.1cm}
	\vfill
	\end{minipage}
        \end{subfigure}
        \hfill
	\begin{subfigure}{.24\linewidth}
		\centering
		\includegraphics[scale = 1]{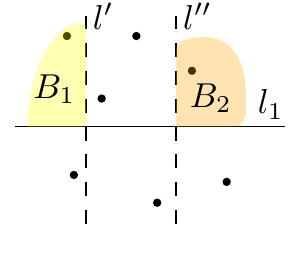}
	\begin{minipage}{.1cm}
	\vfill
	\end{minipage}
	\end{subfigure}
        \hfill
	\begin{subfigure}{.24\linewidth}
		\centering
		\includegraphics[scale = 1]{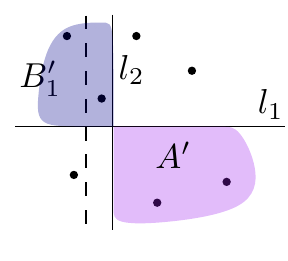}
	\begin{minipage}{.1cm}
	\vfill
	\end{minipage}
	\end{subfigure}	
	\caption{An example of the construction of $p_1$. First the point set $P$ is split by a line $l_1$. Then the lines $l'$ and $l''$ split the point set below $l_1$ into three disjoint parts containing the same number of points, namely $A_1$, $A_2$ and $A_3$. One of $B_1$ and $B_2$ has to contain ''few'' points of $P$, without loss of generality $B_1$, and by slightly changing $l'$ we can ensure that $B_1'$ and $A'$ contain the same number of points of $P$. The resulting lines define $p_1 := l_1 \cap l_2$.}
	\label{fig:boxes2punktein2D}
\end{figure}

Now one of the two outside areas above $l_1$, without loss of generality $B_1$, contains at most $\frac{2n}{7}$ points of $P$. We then move $l'$ slightly towards $l''$, until we have the same number of points in $B_1'$ as in $A'$. We now define $p_1 := l_1 \cap l_2$.

As the area left of $l_2$ and the area below $l_1$ contain $\frac{3n}{7}$ of the points of $P$ every big box contains $p_1$ for sure. On the other hand every small box not containing $p_1$ lies completely above $l_1$ or completely right of $l_2$. By a simple counting argument, any two small boxes not containing $p_1$ intersect. Any small box intersects any big box as a consequence of inequality $(i)$; hence, applying Observation \ref{prop:boxes} we find $p_2$ satisfying the conditions of the Theorem.

For higher dimensions, we use hyperplanes instead of lines and we repeat the second step $d-1$ times (once in every direction except the first).
\end{proof}

A similar spitting idea works for weighted $\eps$-nets of size $3$: Let $l_1$ be a horizontal halving line and let $l_2$ be a vertical halving line. Let $A$ and $B$ be the areas above and below $l_1$ and let $L$ and $R$ be the areas left and right of $l_2$. The lines define four quadrants, where two opposite ones, say $L \cap A$ and $R \cap B$, both contain at least $\frac{n}{4}$ points of $P$. Define $p_1 := l_1 \cap l_2$. For every relevant box $\square$, assign $\square$ to the area $X \in \left\{A,B,L,R\right\}$ for which $|\square \cap X|$ is maximized. Put every box assigned to $A$ and $L$ into $\A$ and every box assigned to $B$ and $U$ into $\B$. Choosing the right values for $\eps_1, \eps_2$ and $\eps_3$, we can apply Observation \ref{prop:boxes} to $\A$ and $\B$ to get the following: 

\begin{theorem}
\label{theorem:boxes3in2D}
Let $P$ be a set of $n$ points in the plane. Let $0 < \eps_1 \leq \eps_2 \leq \eps_3 < 1$ be arbitrary constants with $(i)$ $\eps_1 \geq \frac{3}{8}$,
$(ii)$ $\eps_2 \geq \frac{1}{2}$, and $(iii)$ $\eps_1 + \eps_3 \geq 1$. Then there exist three points $p_1, p_2$ and $p_3$ in $\mathbb{R}^2$ such that every axis-parallel box containing more than $\eps_i n$ points of $P$ contains at least $i$ of the points $p_1, p_2$ and $p_3$.
\end{theorem}

\section{Conclusion}

We have given bounds for weak weighted $\eps$-nets of size $2$ for convex sets and axis-parallel boxes. It remains an interesting question to find bounds for larger sizes. For axis-parallel boxes, we gave a construction for weighted $\eps$-nets of size $3$ in the plane. Unfortunately our construction does not generalize to higher dimensions. It is a natural question whether a similar statement in higher dimensions can be shown using a different construction. 

\bibliographystyle{plainurl}
\bibliography{weighted_eps_nets}

\end{document}